\documentclass[11pt, onecolumn]{article}
\usepackage[top=1in, bottom=1in, left=1.25in, right=1.25in]{geometry}

\usepackage{amsfonts}
\usepackage{amsmath,amssymb}
\usepackage{graphicx}
\usepackage{color, soul}
\usepackage{algorithm,algorithmic}
\usepackage{bm}
\usepackage{booktabs}
\usepackage{flushend}
\usepackage{enumitem}
\usepackage{tikz}
\usetikzlibrary{arrows}
\usepackage{subfigure}

\usepackage[amsmath,thmmarks]{ntheorem}
\usepackage{theorem}

\newtheorem{prop}{Proposition}
\newtheorem{cor}{Corollary}
\newtheorem{lem}{Lemma}
\newtheorem{defi}{Definition}
\newtheorem{rem}{Remark}
\newtheorem{thm}{Theorem}

\theoremheaderfont{\sc}\theorembodyfont{\upshape}
\theoremstyle{nonumberplain}
\theoremseparator{}
\theoremsymbol{\rule{1ex}{1ex}}
\newtheorem{proof}{Proof}
\newtheorem{proof-special}{Proof of Lemma~\ref{lem:slepian}}

\graphicspath {{figures/}}

\newcommand{\probP}{{\mathbb P}}
\newcommand{\probE}{{\mathbb E}}
\newcommand{\probV}{{\mathbb V}}
\newcommand{\mphi}{\bm\Phi}
\newcommand{\vect}[1]{{\mathbf{#1}}}
\newcommand{\ind}{{\mathbf 1}}

\begin{document}

\title{Lower Bound for RIP Constants and Concentration of Sum of Top Order Statistics}

\author{Gen~Li,
Xingyu~Xu,
and~Yuantao~Gu%
\thanks{The authors are with Department of Electronic Engineering, Tsinghua University, Beijing 100084, China. 
The corresponding author of this paper is Y. Gu (gyt@tsinghua.edu.cn).}% <-this % stops a space
}

\date{Manuscript submitted July 5, 2019.}

\maketitle

\begin{abstract}
Restricted Isometry Property (RIP) is of fundamental importance in the theory of compressed sensing  
and forms the base of many exact and robust recovery guarantees in this field. 
Quantitative description of RIP involves bounding the so-called RIP constants 
of measurement matrices. 
In this respect,
it is noteworthy that most results in literature concerning RIP are upper bounds of RIP constants, 
which can be interpreted as theoretical guarantee of successful sparse recovery. 
On the contrary,
the land of lower bounds for RIP constants remains uncultivated. 
Lower bounds of RIP constants, if exist, 
can be interpreted as the \emph{fundamental limit} aspect of successful sparse recovery. 
In this paper, 
the lower bound of RIP constants Gaussian random matrices are derived, 
along with a guide for generalization to sub-Gaussian random matrices. 
This provides a new proof of the fundamental limit 
that the minimal number of measurements needed 
to enforce the RIP of order $s$ is $\Omega(s\log({\rm e}N/s))$, 
which is more straight-forward than the classical Gelfand width argument. 
Furthermore, in the proof we propose a useful technical tool featuring the concentration phenomenon 
for top-$k$ sum of a sequence of i.i.d. random variables, 
which is closely related to mainstream problems in statistics and is of independent interest.

\end{abstract}

\section{Introduction}
Compressed sensing is one of the major achievements in signal processing in the past years. 
The model of compressed sensing can be typically described as 
retrieving some data $\vect x\in\mathbb R^N$ from linear measurements $\vect y=\vect A\vect x$, 
where the measurement matrix $\vect A\in\mathbb R^{n\times N}$ is a underdetermined matrix, i.e. $n<N$. 
Apparently, this task is impossible without proper restrictions on $\vect x$. 
The most common restrictions on $\vect x$ is the sparsity assumption, 
which requires $\vect x$ has at most $s$ non-zero entries, where $s\ll N$. 
If this is the case, $\vect x$ can be efficiently recovered by numerous algorithms, 
for example, the so-called $\ell_1$-minimization:
\begin{equation*}
    \hat{\vect x}=\operatorname{arg~min}_{\vect z}\|\vect z\|_1\text{ s.t. }\vect y=\vect A\vect z.
\end{equation*}

Restricted Isometry Property (RIP) has played a dominant role in analysis of such algorithms
since it was proposed in \cite{Candes2008Restricted}.
The great power of this concept enables researchers to derive 
theoretical guarantees for many popular compressed sensing algorithms
including $\ell_1$-minimization, Orthogonal Matching Pursuit (OMP),
Compressive Sampling Matching Pursuit (CoSaMP), Iterative Hard Thresholding (IHT), 
Hard Thresholding Pursuit (HTP) and so on. 
Proof of such guarantees is based 
on the mode of showing or borrowing some upper bounds of the RIP constants of the sensing matrix 
and then using these upper bounds to analyze specified algorithms. 
This accounts for the extensive study in the literature on upper bounds of RIP constants. 
In contrast, few results on the lower bound of RIP constants are known. 
Similar to the fact that upper bound 
of RIP constants plays the role of success guarantee for sparse recovery, 
lower bound of RIP constants 
could give more insights to the fundamental limit of sparse recovery. 

In this paper we give such a lower bound for Gaussian random matrices 
and provide the method to generalize this result to sub-Gaussian random matrices. 
We also show that the lower bound is tight by proving a new upper bound of RIP constants, 
which is a slight improvement on previous results. 
With this approach we partially rediscover the fact that the minimal number of measurements 
needed to enforce the RIP of order $s$ is $\Omega(s\log({\rm e}N/s))$. 
Compared with the classical proof using Gelfand width, 
our proof bears less generality but is much more straight-forward. 

Furthermore, we identify a useful tool among the lines of our proof, 
which we call by concentration of sum of top order statistics. 
In the literature of probability and statistics, 
asymptotic behaviour and concentration phenomenon of order statistics 
have raised much attention \cite{Boucheron2012Concentration}. 
They are also closely related to the study of empirical process, 
which is one of the mainstream problems in statistics.
On the other hand, sum of top order statistics 
is another classical topic \cite{David2004Order} within this field.
However, known results on joint and individual concentration of order statistics appear 
to yield inequalities for sum of top order statistics that are far from optimal. 
This suggests that more tailored techniques are required to find the optimal bound, 
as we will do in this paper. 
To the best of our knowledge, 
this paper is the first one to establish an exponential concentration inequality for sum of top order statistics. 
We believe this result is of interest in other fields besides compressed sensing. 

We will elucidate the backgrounds and explain our contributions in more technical detail in Section \ref{sec:background}. 
\subsection{Notations}
Bold upper case letters, e.g. $\vect A$, are used to denote matrices, 
while bold lower case letters, e.g. $\vect x$, are used to denote vectors. 
The $s$-th RIP constant of a matrix $\vect A$, to be defined in Section \ref{sec:main}, 
is denoted by $\delta_s(\vect A)$, or simply by $\delta_s$ when there is no confusion. 
$\probP(\cdot)$ is the probability of an event. 
$\probE(\cdot)$ and $\probV(\cdot)$ denotes 
respectively the expectation and the variance of a random variable. 
By convention, a Gaussian matrix is a matrix with i.i.d. standard Gaussian entries. 
$\Psi_s$ denotes the set of non-zero $s$-sparse unit vectors in $\mathbb R^N$, 
c.f. Section \ref{sec:main}.
Denote by $C$ a positive constant that may vary upon each appearance. 
We write $a(n)\sim b(n)$ if $\lim_{n\to\infty}a(n)/b(n)=1$, 
$a(n)\lesssim b(n)$ if $a(n)\le Cb(n)$ for some constant $C$, 
and $a(n)\asymp b(n)$ if both $a(n)\lesssim b(n)$ and $b(n)\lesssim a(n)$ are true.

\subsection{Organization}
The rest of this paper is organized as follows. 
Our main results are presented in Section \ref{sec:main}, 
where we also demonstrate the important corollary on minimal measurements required 
for successful recovery. 
Section \ref{sec:background} is a more involved discussion 
on backgrounds and motivations of these results.
Section \ref{sec:proof-lb}, \ref{sec:proof-unimportant}, and \ref{sec:proof} are devoted to the proof of the main results.
For simplicity, these results are stated for Gaussian matrices, 
but most of them can be generalized easily to sub-Gaussian cases.
We sketch how such generalization could be done in Section \ref{sec:discussion}. 
In Section \ref{sec:simulations}, our new bounds are plotted 
in comparison with previous results. 
Finally in Section \ref{sec:conclusion}, we conclude the paper.

\section{Main Results}\label{sec:main}
We begin with necessary definitions. 
A vector $\vect x$ is called $s$-sparse, or $\|\vect x\|_0\le s$, 
if at most $s$ entries of $\vect x$ are non-zero. 
The $s$-th RIP constant of a matrix $\vect A$ characterizes how close $\vect A$ is to an isometry 
when restricted to the set of $s$-sparse vectors. 
\begin{defi}
The $s$-th RIP constant $\delta_s$ of a matrix $\vect A$ is defined to be the smallest nonnegative number
such that for any $s$-sparse vector $\vect x$, the following holds:
\begin{equation*}
    (1-\delta_s)\|\vect x\|^2\le\|\vect A\vect x\|^2\le(1+\delta_s)\|\vect x\|^2.
\end{equation*}
\end{defi}

RIP constants of random matrices are often asymmetric in the sense that 
the minimal $\delta$ satisfying $(1-\delta)\|\vect x\|^2\le \|\vect A\vect x\|^2$ 
for all $s$-sparse $\vect x$ is essentially different 
from the minimal $\delta$ satisfying $\|\vect A\vect x\|^2\le(1+\delta)\|\vect x\|^2$ 
for all $s$-sparse $\vect x$. 
For this reason we need more intricate notations for RIP constants. 
Denote by $\Psi_s$ the set of non-zero \emph{unit} vectors $\vect x$ in $\mathbb R^N$ such that $\|\vect x\|_0\le s$. 
Define 
\begin{align}
  \delta_s^+=\sup_{\vect x\in\Psi_s}\|\vect A\vect x\|^2-1,\\
  \delta_s^-=1-\inf_{\vect x\in\Psi_s}\|\vect A\vect x\|^2.
\end{align}
It is then obvious that
\begin{equation}
  \delta_s=\max(\delta_s^+,\delta_s^-).
\end{equation}
Hence it suffices to study $\delta_s^+$ and $\delta_s^-$ separately.

The first one of our main results is a lower bound for RIP constants of random matrices. 
For simplicity, we only state and prove the corresponding results for Gaussian matrices.
However, Theorem~\ref{thm:lowerbound} can be easily generalized to sub-Gaussian distributions, 
and Proposition~\ref{prop:gaussian-order-stat} is even more universal. 
Desired readers may find relevant discussions in Section \ref{sec:discussion}.
\begin{thm}\label{thm:lowerbound}
Let $\vect A$ be a Gaussian matrix in $\mathbb{R}^{n\times N}$. 
Then for some constant $C>0$ 
and with probability at least $1-C{\mathrm e}^{-n\varepsilon^2/C}$, 
the RIP constants of $\mathbf{\Phi} = \frac{1}{\sqrt{n}}\vect A$ satisfy
\begin{align*}
\delta_s^+ \ge &\left((1+\sqrt{p}T)(1-\delta)+\frac12pT^2\right)^{1/2}-1, \\
\delta_s^- \ge &1 - \left(1-\sqrt pT+(1+\sqrt pT)\delta+\frac12pT^2\right)^{1/2}, 
\end{align*}
where $s/N<1/5$, $p = \frac{s}{n}$, $\delta=\varepsilon+(n\log \frac{N}{s})^{-1/2}$, 
and $T$ is defined as follows: 
let $Y$ be a $\chi_1^2$ random variable 
and $t$ be the $(1-\frac{s-1}{N-1})$-quantile of $Y$ (i.e. $\probP(Y>t)=\frac{s-1}{N-1})$, 
then $T=\sqrt{\probE(Y\big | Y>t)}$.
\end{thm}

The above lower bound is tight in comparison with the following upper bound: 
\begin{prop}\label{prop:upperbound}
Let $\vect A$ be a Gaussian matrix in $\mathbb{R}^{n\times N}$. 
Then the RIP constants of $\vect\Phi = \frac{1}{\sqrt{n}}\vect A$ satisfy 
\begin{equation*}
  \delta_s = \sqrt{p}T+\frac{C}{\sqrt{n\log \frac{N}{s}}}+\frac{1}{2n}+\varepsilon
\end{equation*}
with probability at least $1-2{\mathrm e}^{-n\varepsilon^2/2}$, 
where $p = \frac{s}{n}<1/5$ and $T$ is defined as follows: 
let $Y$ be a $\chi_1^2$ random variable 
and $t$ be the $(1-\frac{s}{N})$-quantile of $Y$ (i.e. $\probP(Y>t)=\frac{s}{N})$, 
then $T=\sqrt{\probE(Y\big | Y>t)}$.
\end{prop}
\begin{rem}
  The upper bound presented is in fact of the same order as the classical one, 
  see Theorem~\ref{thm:sub-Gaussian-RIP-upper-bound}. 
  The only improvement is a better multiplicative constant. 
  Such improvement, despite being minor, 
  makes the new upper bound very close to the new lower bound given in Theorem~\ref{thm:lowerbound}, 
  demonstrating the tightness of the lower bound. 
\end{rem}

The proof of Theorem~\ref{thm:lowerbound} and Proposition~\ref{prop:upperbound} 
makes crucial use of another main result in this paper, 
namely the concentration of top sum of order statistics. 
Recall that for a sequence of i.i.d. random variables $X_1,\ldots,X_n$, 
the order statistics $X_{(1)},\ldots,X_{(n)}$ are its non-increasing rearrangement. 

\begin{prop}\label{prop:gaussian-order-stat}
Let $Y_1, \ldots, Y_n$ be a sequence of i.i.d. $\chi_1^2$ random variables, 
and $Y_{(1)}, \ldots, Y_{(n)}$ be their order statistics. 
Assume that $k/n<1/5$.
Define
\begin{equation}
T_k = \sqrt{\frac{1}{k}\sum_{i=1}^k Y_{(i)}},
\end{equation}
and 
\begin{equation}
T = \sqrt{\probE  (Y_1 \big| Y_1 > t)},
\end{equation}
where $t$ is defined by $\mathbb{P}\left(Y_1 > t\right) = \frac{k}{n}$, 
then
\begin{equation}
\left|\probE  T_k - T\right| \le \frac{C}{\sqrt{k\log \frac{n}{k}}},
\end{equation}
where $C$ is a positive constant, 
and
\begin{equation} 
    \probP\left( \left|T_k-T\right| > \frac{C}{\sqrt{k\log \frac{n}{k}}} + \varepsilon \right) \le 2{\rm e}^{-\frac{k\varepsilon^2}{2}}, \quad \forall \varepsilon>0.
\end{equation}
\end{prop}

We briefly discuss some corollaries of these results, 
of which the most important one is the minimal number of measurements for successful recovery. 
First we need an asymptotic estimation of $t$ and $T$ in Theorem~\ref{thm:lowerbound}.
\begin{prop}\label{prop:asymp}
In the asymptotic regime $N\to\infty$, $s/N\to0$, we have
\begin{equation}
  t\sim2\log\frac Ns,\quad 
  T\sim\sqrt{2\log\frac Ns}.
\end{equation}
As a consequence, in the asymptotic regime $N\to\infty$, 
$s/N\le\gamma$ for some constant $\gamma\in(0,1/5)$, 
we have
\begin{equation*}
  t\asymp\log\frac Ns,\quad 
  T\asymp\sqrt{\log\frac Ns}.
\end{equation*}
\end{prop}
\begin{proof}
See Section \ref{app:asymp}.
\end{proof}

\begin{cor}[Minimal number of measurements]\label{cor:no.-of-measurements}
  The minimal number of random Gaussian measurements to enforce RIP is $\Omega(s\log(N/s))$.
\end{cor}
\begin{proof}
  By Theorem~\ref{thm:lowerbound}, the RIP is enforced only if $pT^2$ is bounded by some constant. 
  From the foregoing asymptotic analysis we see that $T\asymp\sqrt{\log\frac Ns}$. 
  Thus $n=s/p=sT^2/(pT^2)=\Omega(s\log\frac Ns)$.
\end{proof}

\section{Related Works}\label{sec:background}
\subsection{Restricted Isometry Property}
In the literature of compressed sensing, 
RIP is a powerful tool to prove exact and robust recovery results for various algorithms. 
Here exact recovery means that the algorithm recovers $\vect x$ exactly 
from the measurements $\vect y=\vect A\vect x$ for all $s$-sparse $\vect x$, 
and robust recovery means that the algorithm returns 
from noisy measurements $\vect y=\vect A\vect x+\vect e$ an estimate $\tilde{\vect x}$ of $\vect x$ with accuracy
\[\|\vect x-\tilde{\vect x}\|_2\le C_1\frac{\sigma_s(\vect x)_1}{\sqrt s}+C_2\tau\]
for \emph{all} vectors $\vect x$, 
where $C_1$, $C_2$ are positive constants, 
$\sigma_s(\vect x)_p$ is the $\ell_p$-compressibility of $\vect x$, 
defined as 
\[\sigma_s(\vect x)_p=\inf_{\|\vect z\|_0\le s}\|\vect x-\vect z\|_p,\]
and $\tau$ is a deterministic upper bound of the noise level 
satisfying $\tau\ge\|\vect e\|_2$. 
We partly borrow from \cite{Rauhut2012Restricted} the following 
Table~\ref{tab:RIP-and-algorithm} on the state-of-art results 
of RIP requirements for various algorithms 
to ensure successful exact and robust recovery for $s$-sparse vectors. 
All of these results can be found, for instance, in \cite{Foucart2017Mathematical} 
and the references therein. 

\begin{table}[t]
\centering
\caption{Requirements on RIP constants for various recovery algorithms \cite{Rauhut2012Restricted}.}
\renewcommand\arraystretch{1.2}
\begin{tabular}{cc}
\hline
Algorithm & Requirements\\
\hline
$\ell_1$-minimization & $\delta_s<\frac13$ or $\delta_{2s}<\frac4{\sqrt{41}}$\\
OMP & $\delta_s<\frac{1}{1+\sqrt s}$ or $\delta_{13s}<\frac16$ \\
CoSaMP & $\delta_{4s}\le\frac{\sqrt{\sqrt{11/3}-1}}{2}$ \\
IHT, HTP & $\delta_{3s}<\frac{1}{\sqrt3}$\\
\hline
\end{tabular}
\label{tab:RIP-and-algorithm}

\renewcommand\arraystretch{1}
\end{table}

Ideally, if one may effectively evaluate the RIP constants of a given matrix $\vect A$, 
not much is left to worry about: with the precise value of RIP constants
it can be inferred from Table~\ref{tab:RIP-and-algorithm} 
whether we have theoretical guarantee for any of our favourite algorithms. 
Unfortunately, 
this is impossible due to the NP-hardness of evaluating RIP constants \cite{tillmann2013computational}. 
In fact, even qualitatively certifying RIP is NP-hard \cite{bandeira2013certifying}. 
A common substitute for the precise values of RIP in theoretical anlysis 
is the \emph{upper bound} of RIP constants, 
which would suffice to guarantee successful recovery. 
The RIP upper bound for sub-Gaussian matrices, stated as below, is a landmark of such results.
\begin{thm}\emph{\cite{baraniuk2008simple}}\label{thm:sub-Gaussian-RIP-upper-bound}
Let $\vect A$ be a $n\times N$ random matrix with i.i.d. (centered) sub-Gaussian entries 
of variance $1$ and sub-Gaussian norm uniformly bounded by $K>0$. 
Then the RIP constants of $\mphi\overset{\triangle}{=}\frac{1}{\sqrt n}\vect A$ satisfies
\begin{equation}\label{eqn:RIP-upper-bound-prob}
    \probP(\delta_s>\delta)\le2\binom Ns \exp(-c_1\delta^2n+c_2s)
\end{equation}
where $c_1$, $c_2$ are some positive constants that depend only on $K$.
\end{thm}
\begin{cor}
For measurement matrix $\mphi$ as in Theorem~\ref{thm:sub-Gaussian-RIP-upper-bound}, 
there exists some positive constant $C$ that depends only on $K$, such that 
$\delta_s\le\delta$ holds with probability at least $1-\varepsilon$ whenever
\begin{equation}\label{eqn:classical-sub-Gaussian-bound}
  n\ge C\delta^{-2}(s\log({\rm e}N/s)+\log(2\varepsilon^{-1})).
\end{equation}
\end{cor}

The above corollary indicates that when the number of measurements $n\gtrsim s\log(eN/s)$, 
successful recovery is guaranteed with high (in fact, $1-2\exp(-C'\delta^{-2}m)$) probability. 

We turn to the opposite side of the problem: 
what is the minimal number of measurements 
when successful recovery is to be expected? 
A classical argument by estimating Gelfand width \cite{foucart2010gelfand} solves this problem 
by showing that for \emph{any} measurement matrix $\vect A\in\mathbb R^{n\times N}$ 
and \emph{any} recovery algorithm $\Delta$, interpreted as a map $\mathbb R^n\to\mathbb R^N$, 
if $\|\vect x-\Delta(\vect A\vect x)\|_2\le C\sigma_s(\vect x)_1$ holds for all $\vect x$, 
we necessarily have 
\[n\gtrsim s\log\left(\frac{{\rm e}N}{s}\right)\]

This, when combined with Table~\ref{tab:RIP-and-algorithm}, 
yields the following corollary (in a stronger form).
\begin{cor}
If the $s$-th RIP constant of $\vect A\in\mathbb R^{n\times N}$ satisfies 
$\delta_s<\frac{1}{\sqrt 3}$, then 
\begin{equation}\label{eqn:Gelfand-bound}
    n\ge cs\log\left(\frac{{\rm e}N}{s}\right)
\end{equation}
for some constant $c>0$ depending only on $\delta_s$.
\end{cor}

Taking $\varepsilon=1/2$ in \eqref{eqn:classical-sub-Gaussian-bound} 
and comparing \eqref{eqn:classical-sub-Gaussian-bound} and \eqref{eqn:Gelfand-bound}, 
one readily checks that the bound in \eqref{eqn:classical-sub-Gaussian-bound} is optimal 
up to a multiplicative constant. 
These pieces when put together constitutes an almost-conclusive answer 
to the problem of minimal number of measurements required for successful recovery. 
However, two issues remain unsolved: 
i) Theorem~\ref{thm:sub-Gaussian-RIP-upper-bound} 
is a standard application of random matrix theory, 
which is a tool intimate for compressed sensing society, 
while the proof using Gelfand width is indirect and even involves analysis of $\ell_1$-minimization algorithm; 
ii) \eqref{eqn:RIP-upper-bound-prob} gives a probabilistic bound of RIP constants, 
while \eqref{eqn:Gelfand-bound} is deterministic. 
We therefore pose the following questions: 
\emph{
\begin{itemize}
\item Is there a more direct (i.e. random-matrix-theoretic) proof of \eqref{eqn:Gelfand-bound}?
\item Can we find a probabilistic lower bound of RIP constants, 
possibly by bounding the probability $\probP(\delta_s\le\delta)$ for small $\delta$?
\end{itemize}
}
Theorem~\ref{thm:lowerbound} is an affirmative answer to these questions.

\subsection{Concentration of Order Statistics}
Distribution of order statistics is a well-investigated topic; 
the well-known R\'enyi's representation provides an explicit formula for their distribution function \cite{David2004Order}. 
On the other hand, 
concentration of order statistics is still an active field of research \cite{Boucheron2012Concentration}. 
Many researches in this vein were inspired by the concentration of measure phemonenon \cite{ledoux2001concentration}. 
However, general principles in concentration of measure theory do not supply 
satisfactory bounds for order statistics. 
For example, it is known that $\probV(X_{(\lfloor n/2\rfloor)})=O(1/n)$, 
while the powerful logarithmic Sobolev inequality in concentration of measure theory implies 
only $\probV(X_{(\lfloor n/2\rfloor)})\le 1$. 
Additional efforts are in need to establish tight bounds for concentration of order statistics. 

The most notable results on concentration of order statistics 
are a series of inequalities stemming from Glivenko-Cantelli theorem. 
We record here one of such inequalities for reference.
\begin{thm}[Dvoretzky–-Kiefer–-Wolfowitz inequality\footnotemark, \cite{massart1990tight}]\label{thm:DKW}
Let $X_1,\ldots, X_n$ be i.i.d. random variables 
with cumulative distribution function $F(x)$. 
Denote by $F_n$ the associated empirical distribution function, defined by 
\begin{equation*}
  F_n(x)=\frac1n\sum_{i=1}^n\ind_{\{X_i\le x\}}.
\end{equation*}
Then for any $\varepsilon>0$ the following holds:
\begin{equation*}
  \probP\left(\sup_{x\in\mathbb R}|F_n(x)-F(x)|>\varepsilon\right)\le 2\exp(-2n\varepsilon^2).
\end{equation*}
\end{thm}
\footnotetext{
A weaker form of this theorem is well-known to statisticians 
as Kolmogorov-Smirnov test. 
}

In probability theory, 
quantities in the form of $\ind_{\{X\le x\}}$ are (heuristically) considered to be of zero-th order, 
in contrast with first (second, third, \ldots) order quantities such as $X$ ($X^2$, $X^3$, \ldots).
One may thus regard Theorem~\ref{thm:DKW} as a zero-th order joint concentration theorem 
for order statistics. 
Concentration theorem in higher order for any individual 
of order statistics is also available in the literature, c.f. \cite{Boucheron2012Concentration}, 
but is described in a rather complicated form 
which does not imply convenient tail bound. 
On the other hand, first order joint concentration theorem, 
i.e. a tail bound on concentration of $\sum_{i=1}^k X_{(i)}$, is almost unknown. 
Our result may hopefully fill this vacancy.

% ---------------------------------------------------------
\section{Proof of Theorem~\ref{thm:lowerbound}}\label{sec:proof-lb}
For convenience we first set up some notations. 
Let 
\begin{align*}
  \sigma^s_{\max}(\vect A)&=\sup_{\vect x\in\Psi_s}\|\vect A\vect x\|,\\
  \sigma^s_{\min}(\vect A)&=\inf_{\vect x\in\Psi_s}\|\vect A\vect x\|.
\end{align*}
By definition we have 
\begin{align}
  \delta_s^+(\mphi)&=\frac{1}{\sqrt n}\sigma_{\max}^s(\vect A)-1,\label{eqn:rip-max-sing-val}\\
  \delta_s^+(\mphi)&=1-\frac{1}{\sqrt n}\sigma_{\min}^s(\vect A).\label{eqn:rip-min-sing-val}
\end{align}

In this section we assume that $\vect A$ is a random matrix with i.i.d. standard Gaussian entries.
For an $N$-dimensional unit vector $\vect v=(v_1,v_2,\ldots,v_N)$, we have
\begin{equation}\label{eqn:Av}
    \|\vect A\vect v\|^2=v_1^2\|\vect a_1\|^2+2v_1\sum_{i=2}^Nv_i\langle \vect a_1,\vect a_i\rangle+\sum_{i,j=2}^Nv_iv_j\langle \vect a_i,\vect a_j\rangle.
\end{equation}

Set 
\begin{align} 
x_i&=\frac{\langle\vect a_1,\vect a_i\rangle}{\|\vect a_1\|},\nonumber\\
\vect b_i&=\vect a_i-x_i\frac{\vect a_1}{\|\vect a_1\|},\label{eq:definebi}
\end{align}
for $i=2,\ldots,N$. 
By a well-known property of Gaussian distribution, 
$\{\vect a_1,x_2,\ldots,x_N\}$ are jointly independent and are all standard Gaussian vectors/variables.
Conditioning on $\vect a_1$, we see that $$\{x_2,\ldots,x_N,\vect b_2,\ldots,\vect b_N\}$$ are jointly independent, 
and $\vect b_i$'s are Gaussian vectors with covariance matrix $${\bf\Sigma}={\bf I}-\frac{\vect a_1\vect a_1^{\rm T}}{\|\vect a_1\|^2}.$$ 

With \eqref{eq:definebi} and the fact that $\langle\vect a_1,\vect b_i\rangle=0$, \eqref{eqn:Av} can be expressed as
\begin{align}\label{eqn:Av-final}
    \|\vect A\vect v\|^2=&v_1^2\|\vect a_1\|^2+2v_1\|\vect a_1\|\sum_{i=2}^Nv_ix_i
    +\sum_{i,j=2}^Nv_iv_j(x_ix_j+\langle\vect b_i,\vect b_j\rangle)\nonumber\\
    =&v_1^2\|\vect a_1\|^2+2v_1\|\vect a_1\|\sum_{i=2}^Nv_ix_i+\left(\sum_{i=2}^Nv_ix_i\right)^2
    +\sum_{i,j=2}^Nv_iv_j\langle\vect b_i,\vect b_j\rangle.
\end{align}

We are interested in extremal values of $\|\vect A\vect v\|^2$ 
as $\vect v$ ranges over the set of $s$-sparse unit vectors. 
To establish the lower bound in Theorem~\ref{thm:lowerbound}, 
we will designate some specific values of $\vect v$ to estimate these extremal values. 
That is, for any specific choice of $s$-sparse unit vectors $\vect v'$, 
we have 
\begin{equation}
  (\sigma_{\max}^s(\vect A))^2\ge\|\vect A\vect v'\|^2,
\end{equation}
and
\begin{equation}
  (\sigma_{\min}^s(\vect A))^2\le\|\vect A\vect v'\|^2.
\end{equation}
In case that $\|\vect A\vect v'\|^2$ is sufficiently close to $\max_{\vect v\in\Psi_s}\|\vect A\vect v\|^2$ 
(resp. $\min_{\vect v\in\Psi_s}\|\vect A\vect v\|^2$)
and is easy to compute, 
the above method will provide a satisfactory lower bound for $\max_{\vect v\in\Psi_s}\|\vect A\vect v\|^2$ 
(resp. $\min_{\vect v\in\Psi_s}\|\vect A\vect v\|^2$).

Next we show how to construct such $\vect v'$. 
Take $v_1'=1/\sqrt2$. 
By the proof of Cauchy-Schwarz inequality, 
there exists a suitable choice of $(v_2',\ldots,v_N')$ which is $(s-1)$-sparse and 
fulfills 
\begin{align*}
    \sum_{i=2}^Nv_i'^2&=\frac12,\\
    \sum_{i=2}^Nv_i'x_i&=\sqrt{\sum_{i=2}^Nv_i'^2}\sqrt{\sum_{j=1}^{s-1}x_{(j)}^2}.
\end{align*}

Moreover, such choice makes $\vect v'$ a $(x_2,\ldots,x_N)$-measurable random vector. 
Combining the above equations, we have
\begin{align}
    \|\vect A\vect v'\|^2=&\frac12\|\vect a_1\|^2+\|\vect a_1\|\sqrt{\sum_{j=1}^{s-1}x_{(j)}^2}+\frac12\sum_{j=1}^{s-1}x_{(j)}^2
    +\sum_{i,j=2}^Nv_i'v_j'\langle\vect b_i,\vect b_j\rangle.
\end{align}

On the other hand, 
taking $$\vect v''=\left[v_1',-v_2',-v_3',\ldots,-v_N'\right]^{\rm T},$$ 
we have
\begin{align}
    \|\vect A\vect v''\|^2=&\frac12\|\vect a_1\|^2-\|\vect a_1\|\sqrt{\sum_{j=1}^{s-1}x_{(j)}^2}+\frac12\sum_{j=1}^{s-1}x_{(j)}^2
    +\sum_{i,j=2}^Nv_i'v_j'\langle\vect b_i,\vect b_j\rangle.
\end{align}

A lower bound (resp. upper bound) of $\|\vect A\vect v'\|^2$ (resp. $\|\vect A\vect v''\|^2$), 
hence of $\max_{\vect v\in\Psi_s}\|\vect A\vect v\|^2$ (resp. $\min_{\vect v\in\Psi_s}\|\vect A\vect v\|^2$), 
is obtained immediately as a consequence 
of Proposition~\ref{prop:gaussian-order-stat} 
and standard concentration inequalities for Gaussian quadratic forms. 
In fact, from Bernstein inequality we see that for some universal constant $C>0$ 
and any $\varepsilon>0$, the following holds:
$$
    \probP\left(\frac1n\Big|\|\vect a_1\|^2-n\Big|>\varepsilon\right)
    \le2\exp(-Cn\min(\varepsilon,\varepsilon^2)).
$$
Similarly, setting 
\begin{equation*}
    \Gamma_{\varepsilon}\overset{\triangle}{=}\left\{\frac1n\left|\sum_{i,j=2}^Nv_i'v_j'\langle \vect b_i,\vect b_j\rangle-(n-1)(1-v_1^2)\right|>\varepsilon\right\},
\end{equation*}
and by conditional independence of $\vect v'$ and $\vect b_i$, we have
\begin{align*}
    \probP\left(\Gamma_{\varepsilon}\Big|\vect a,v_1',\ldots,v_n'\right)
    \le&2\exp\left[-Cn\min\left(\frac{\varepsilon}{1-v_1^2},\frac{\varepsilon^2}{(1-v_1^2)^2}\right)\right]\nonumber\\
    \le&2\exp(-Cn\min(\varepsilon,\varepsilon^2)),
\end{align*}
which follows from Hanson-Wright inequality \cite{rudelson2013hanson}.
Thus by integrating we obtain
\begin{equation}\label{eqn:Av_minor_quadratic_term}
    \probP(\Gamma_{\varepsilon})\le 2\exp(-Cn\min(\varepsilon, \varepsilon^2)).
\end{equation}

Applying Proposition~\ref{prop:gaussian-order-stat} to $\sum_{j=1}^{s-1}x_{(j)}^2$ 
and taking into account the foregoing arguments, 
the conclusion of Theorem~\ref{thm:lowerbound} follows immediately.

% -----------------------------------------------------------------------------
\section{Proof of Proposition~\ref{thm:lowerbound}}\label{sec:proof-unimportant}

The proof is a slight modification of the standard one. 
From the definition of $\sigma_{\max}^s$ and $\sigma_{\min}^s$,
it is readily verified that
\begin{prop} \label{Lip_ineq}
$\sigma_{\max}^s(\vect A), \sigma_{\min}^s(\vect A)$ is $1$-Lipschitz in $\vect A$. 
In other words, for $\vect A_1,\vect A_2 \in\mathbb{R}^{n\times N}$, we have
\begin{gather}
\sigma_{\max}^s(\vect A_1) - \sigma_{\max}^s(\vect A_2)\le \left\| \vect A_1-\vect A_2 \right\|_{\rm F},\label{eqn:lipschitz1} \\
\sigma_{\min}^s(\vect A_1) - \sigma_{\min}^s(\vect A_2)\le \left\| \vect A_1-\vect A_2 \right\|_{\rm F}.\label{eqn:lipschitz2}
\end{gather}
where $\|\cdot\|_{\rm F}$ denotes the Frobenius norm. 
In fact, \eqref{eqn:lipschitz1} and \eqref{eqn:lipschitz2} hold even if Frobenius norm is replaced by operator norm.
\end{prop}

By Proposition~\ref{Lip_ineq} and concentration of measure (Appendix \ref{app:tool}), we have
\begin{gather}
  \probP\left(\sigma_{\max}^s(\vect A) > \probE\sigma_{\max}^s(\vect A)+\varepsilon \right) \le {\rm e}^{-\frac{\varepsilon^2}{2}},\label{eqn:concentration-max-sing-val} \\
  \probP\left(\sigma_{\min}^s(\vect A) < \probE\sigma_{\min}^s(\vect A)-\varepsilon \right) \le {\rm e}^{-\frac{\varepsilon^2}{2}}.\label{eqn:concentration-min-sing-val}
\end{gather}

To prove Proposition~\ref{prop:upperbound}, 
it suffices to bound $\probE\sigma_{\max}^s(\vect A)$ 
and $\probE\sigma_{\min}^s(\vect A)$.
For such purpose we need a well-known comparison lemma for Gaussian process 
depicted in Appendix \ref{app:tool}, Lemma~\ref{lem:slepian-fernique}. 
Specifically, we shall use the following consequence of Lemma~\ref{lem:slepian-fernique}.
\begin{lem}\label{lem:slepian}
Let ${\mathbf g}, {\mathbf h}$ be standard Gaussian random vectors. 
Define two Gaussian processes $X_{\vect u,\vect v}$ and $Y_{\vect u,\vect v}$ 
on the set $\{\vect u\in\mathbb S^{n-1},\vect v\in\Psi_k\}$ as following:
\begin{align*}
X_{\vect u,\vect v} &= \left\langle \vect A\vect u,\vect v \right\rangle, \\
Y_{\vect u,\vect v} &= \left\langle {\vect g},\vect u \right\rangle + \left\langle {\vect h},\vect v \right\rangle,
\end{align*}
then we have
\begin{align*}
\probE\max_{\vect v \in \Psi_k}\max_{\vect u \in \mathbb{S}^{n-1}} X_{\vect u,\vect v} \le \probE\max_{\vect v \in \Psi_k}\max_{\vect u \in \mathbb{S}^{n-1}}Y_{\vect u,\vect v}, \\
\probE\min_{\vect v \in \Psi_k}\max_{\vect u \in \mathbb{S}^{n-1}} X_{\vect u,\vect v} \ge \probE\min_{\vect v \in \Psi_k}\max_{\vect u \in \mathbb{S}^{n-1}}Y_{\vect u,\vect v}.
\end{align*}
\end{lem}
\begin{proof}
  See Section \ref{app:tool}.
\end{proof}

\subsection{Upper Bound of $\probE\sigma_{\max}^s(\vect A)$}

By Lemma~\ref{lem:slepian}, we have
\begin{align}
\probE\sigma_{\max}^s(\vect A)=&\probE\max_{\vect v \in \Psi_s}\max_{\vect u \in \mathbb{S}^{n-1}}\vect u^{\rm T} \vect A \vect v\nonumber \\
\le& \probE  \max_{\vect u \in \mathbb{S}^{n-1}}  \left\langle {\mathbf g},\vect u \right\rangle + \probE  \max_{\vect v \in \Psi_s} \left\langle {\mathbf h},\vect v \right\rangle\nonumber \\
=& \probE \left\|{\vect g}\right\| + \probE \left\|{\vect h}_k^\sharp\right\|\label{eqn:slepian-upper}
\end{align}
where $$\vect h_k^\sharp=[h_{(1)},\ldots,h_{(k)},0,\ldots,0]^{\rm T}\in\mathbb{R}^{N}$$
is a $k$-sparse vector obtained from $\vect h$ 
by expunging all entries of $\vect h$ except those $k$ entries with largest magnitudes. 
Here $h_{(i)}$ denotes the entry in $\vect h$ with $i$-th largest magnitude. 
Note that $h_{(i)}^2$ can be viewed as the $i$-th order statistic of i.i.d. 
$\chi_1^2$ random variables $h_1^2,\ldots,h_n^2$. 
Thus $$\probE\|\vect h_k^\sharp\|=\probE\sqrt{\sum_{i=1}^k h_{(i)}^2}$$ can be bounded 
with Proposition~\ref{prop:gaussian-order-stat}. 
We have 
\begin{equation}\label{eqn:prop1-ord-bound}
  \probE\sqrt{\sum_{i=1}^k h_{(i)}^2}\le\sqrt{k}T + \frac{C}{\sqrt{\log \frac{N}{k}}}.
\end{equation}

On the other hand, it is well-known that
\begin{equation}
  \probE\|\vect g\|=\frac{\sqrt 2\Gamma(\frac{n+1}2)}{\Gamma(\frac n2)},
\end{equation}
and by Kazarinoff's inequality for binomial coefficients \cite{qi2012bounds} (or simply by Cauchy-Schwarz) we have
\begin{equation}
  \probE\|\vect g\|\le\sqrt n.
\end{equation}

Thus 
\begin{equation}\label{eqn:max-sing-val-expectation}
  \probE\sigma_{\max}^s(\vect A)\le\sqrt n+\sqrt{k}T + \frac{C}{\sqrt{\log \frac{N}{k}}}.
\end{equation}

\subsection{Lower Bound of $\probE\sigma_{\min}^s(\vect A)$}

Similar to \eqref{eqn:slepian-upper}, one may invoke Lemma~\ref{lem:slepian} to obtain
\begin{align}
\probE\sigma_{\min}^s(\vect A)=&\probE\min_{\vect v \in \Psi_s}\max_{\vect u \in\mathbb{S}^{n-1}}\vect u^{\rm T} \vect A \vect v\nonumber \\
\ge&\probE\max_{\vect u \in \mathbb{S}^{n-1}}\left\langle {\mathbf g},\vect u \right\rangle + \probE\min_{\vect v \in \Psi_s} \left\langle {\mathbf h},\vect v \right\rangle\nonumber \\
=& \probE \left\|{\vect g}\right\| - \probE \left\|{\vect h}_k^\sharp\right\|.\label{eqn:prop1-min-sing-val}
\end{align}
Now by Kazarinoff's inequality we have 
\begin{equation}\label{eqn:Kazarinoff-lb}
\probE \left\|{\vect g}\right\|\ge \sqrt{n-\frac12}.
\end{equation}
Combining \eqref{eqn:prop1-ord-bound}, \eqref{eqn:prop1-min-sing-val}, and \eqref{eqn:Kazarinoff-lb}, 
we have
\begin{equation}\label{eqn:min-sing-val-expectation}
\probE\sigma_{\min}^s(\vect A)\ge\sqrt{n-\frac12}-\sqrt{k}T - \frac{C}{\sqrt{\log \frac{N}{k}}}.
\end{equation}
Note that $\sqrt{n-\frac12}\ge\sqrt n-\frac{1}{2\sqrt n}$. The conclusion of Proposition~\ref{prop:upperbound} then follows from 
\eqref{eqn:rip-max-sing-val}, \eqref{eqn:rip-min-sing-val}, 
and
\eqref{eqn:concentration-max-sing-val}, \eqref{eqn:concentration-min-sing-val}, \eqref{eqn:max-sing-val-expectation}, \eqref{eqn:min-sing-val-expectation}.

% -----------------------------------------------------------------------------

\section{Concentration of order statistics}\label{sec:proof}

In this section we will prove Proposition~\ref{prop:gaussian-order-stat}. 
The quantity of interest here is the sum $\sum_{i=1}^{k}X_{(i)}$, 
where $X_{(1)}\ge\ldots\ge X_{(n)}$ are the order statistics of i.i.d. non-negative random variables 
$X_1,\ldots,X_n$. 
We begin with a concentration inequality for $X_{(k)}$ 
which can be regarded as a local version of Theorem~\ref{thm:DKW}. 
\begin{lem}\label{lem:ord_concentration}
Let $X_1, \ldots, X_n$ be a sequence of i.i.d. non-negative random variables 
with distribution function $F$, 
and $X_{(1)}, \ldots, X_{(n)}$ be theirs order statistics. 
Denote $\alpha=k/n<1/2$. 
Assume $t^-$, $t^+$ are positive real numbers such that
\begin{gather*}
  \probP(X_1>t^-)=\alpha+\delta,\\
  \probP(X_1>t^+)=\alpha-\delta,
\end{gather*}
where $\delta$ is a small positive constant, e.g. $\delta\in(1-\alpha,\alpha)$.
Then 
\begin{gather}
  \probP(X_{(k+1)}\le t^-)\le\exp\left(-n\delta^2\left(\frac1\alpha+\frac1{1-\alpha}\right)\log\frac{\rm e}2\right),\label{eqn:ord_lowerbound} \\
  \probP(X_{(k-1)}>t^+)\le\exp\left(-n\delta^2\left(\frac1\alpha+\frac1{1-\alpha}\right)\log\frac{\rm e}2\right).\label{eqn:ord_upperbound} 
\end{gather}
In particular, since $X_{(k+1)}\le X_{(k)}\le X_{(k-1)}$, we have
\begin{align*}
  \probP(t^-<X_{(k)}\le t^+)
  \ge1-2\exp\left(-n\delta^2\left(\frac1\alpha+\frac1{1-\alpha}\right)\log\frac{\rm e}2\right).
\end{align*}
\end{lem}
\begin{proof}
Observe that $X_{(k+1)}\le t^-$ is equivalent to $\sum_{i=1}^n\ind_{\{X_i>t^-\}}\le k.$ 
Since $\ind_{\{X_i>t^-\}}$ are i.i.d. Bernoulli variables with 
\begin{align*}
  \probP(\ind_{\{X_i>t^-\}}=1)&=\probP(X_i>t^-)\\ &=\alpha+\delta,
\end{align*}
we see that $\sum_{i=1}^n\ind_{\{X_i>t^-\}}$ follows binomial distribution $B(n,\alpha+\delta)$. 
Thus one may apply classical entropy-type tail bounds (see \cite{cover2012elements} for example) 
to obtain 
\begin{equation}\label{eqn:lem_ord_concentration_1}
  \probP\left(\sum_{i=1}^n\ind_{\{X_i>t^-\}}\le k\right)\le\exp(-nD(\alpha\Vert \alpha+\delta)),
\end{equation}
where $D(\alpha\Vert\alpha+\delta)$ is the relative entropy given by 
\begin{equation*}
  D(\alpha\Vert\alpha+\delta)=\alpha\log\frac{\alpha}{\alpha+\delta}+(1-\alpha)\log\frac{1-\alpha}{1-\alpha-\delta}.
\end{equation*}
Equivalently, we have
\begin{equation*}
  D(\alpha\Vert\alpha+\delta)=-\alpha\log\left(1+\frac{\delta}{\alpha}\right)-(1-\alpha)\log\left(1-\frac{\delta}{1-\alpha}\right). 
\end{equation*}
For $x\in(-1,1)$, we have the following elementary inequality:
\begin{equation*}
  \log(1+x)\le x-\left(\log\frac {\rm e}2\right)x^2,
\end{equation*}
hence 
\begin{equation}\label{eqn:lem_ord_concentration_2}
  D(\alpha\Vert\alpha+\delta)\ge\delta^2\left(\frac1{\alpha}+\frac{1}{1-\alpha}\right)\log\frac {\rm e}2.
\end{equation}

Plugging (\ref{eqn:lem_ord_concentration_2}) into (\ref{eqn:lem_ord_concentration_1}) 
completes the proof of (\ref{eqn:ord_lowerbound}). 
The proof of (\ref{eqn:ord_upperbound}) is verbatim.
\end{proof}
\begin{rem}
  The quantile $t^-$ and $t^+$ in Lemma~\ref{lem:ord_concentration} does exist for any (reasonable) $\delta$
when the distribution $F$ is continous on $[0,\infty)$. 
Assume further that $F$ is absolutely continous on $[0,\infty)$ and $f$ is its density, 
then 
\begin{equation}\label{eqn:deviation-and-density}
  \probP(|X_{(k)}-t|>\varepsilon)\le 2\exp\left(-Kn\delta^2(\alpha^{-1}+(1-\alpha)^{-1})\right),
\end{equation}
where $t$ is such that $$\probP(X_1>t)=1-\alpha,$$ and $$K=\left(\log\frac {\rm e}2\right)\inf_{|x-t|\le\varepsilon}f^2(x).$$
\end{rem}

The proof proceeds as following. 
Note that
$$
|T_k-T| \le |T_k - \probE  T_k| + \left|\probE  T_k - \sqrt{\probE  T_k^2}\right| + \left|\sqrt{\probE  T_k^2} -T\right|.
$$
We will analyze the three terms in the right hand side 
to prove the required concentration inequality for $|T_k-T|$.
The bound of $|\probE T_k-T|$ is obtained as a byproduct.

It follows from rearrangement inequality that $T_k$ is a $\sqrt{\frac1k}$-Lipschitz function 
in $\vect X=(X_1,\ldots,X_n)$. 
By concentration of measure (Appendix \ref{app:tool}), we have
\begin{equation}\label{eqn:T-concentration}
  \probP(|T_k-\probE T_k| > \varepsilon) \le 2{\rm e}^{-\frac{k\varepsilon^2}{2}}, \quad \forall \varepsilon>0.
\end{equation}
and consequently
\begin{equation}
  \probP(|T_k^2-(\probE T_k)^2|>t^2+2t\probE T_k)\le 2{\rm e}^{-\frac{kt^2}2}.
\end{equation}

To bound $\probE T_k-\sqrt{\probE T_k^2}$, we instead inspect $\probE T_k^2-(\probE T_k)^2$.
Let 
\begin{equation*}
  f(x) = \frac{1}{\sqrt{2\pi x}}{\rm e}^{-\frac{x}{2}}
\end{equation*}
and $F(x)$ be the p.d.f. and the c.d.f. of $\chi_1^2$ distribution respectively. 

Observe that
\begin{align}
\left|\probE T_k^2 - (\probE T_k)^2\right| \le& \probE \left|T_k^2 - (\probE T_k)^2\right|\nonumber \\
=& \int_{0}^{\infty}\mathbb{P}\left(\left|T_k^2 - (\probE T_k)^2\right| > x\right)\mathrm{d}x\nonumber \\
\le& \int_{0}^{\infty} \mathbb{P}\left(\left|T_k^2 - (\probE T_k)^2\right| > x\right)\mathrm{d}x\nonumber \\
\le& \int_{0}^{\infty} 2{\rm e}^{-\frac{kx^2}{2}}\mathrm{d}(x^2+2x\probE T_k)\nonumber \\
=&\frac4k + \frac{2\sqrt{2\pi}\probE T_k}{\sqrt{k}}.
\end{align}

Next, we bound the difference between $\probE T_k^2$ and $T^2$.

Given $X_{(k + 1)}^2$, $T_k^2$ can be seen as a function of $k$ i.i.d. random variable $Y_1, \ldots, Y_k$ with $Y_i \sim X^2\big| X^2 > X_{(k + 1)}^2$, 
where $T_k^2 = \frac{1}{k}\sum_{i=1}^k Y_k$ and $X^2 \sim \chi_1^2$,
then the conditional expectation of $T_k^2$ is
\begin{align}
\probE\left(T_k^2 \big| X_{(k + 1)}^2\right)&= \probE  \frac{1}{k}\sum_{i=1}^k Y_i\nonumber \\
&= \probE  Y_1\label{eqn:ord-stat-condition-expectation}.
\end{align}
We inspect the quantity $\probE(X^2|X^2>s)$. By straight-forward computation, we have
\begin{align}
\probE \left(X^2 \big| X^2>s\right) 
=& \frac{\int_{s}^{\infty}xf(x)\mathrm{d}x}{\int_{s}^{\infty}f(x)\mathrm{d}x}\nonumber\\
=& \frac{s}{\probE\sqrt{\frac{s}{s+Y}}}+1,\label{eqn:chi-square-conditional-expectation} 
\end{align}
where $Y\sim\mathrm{Exp}(\frac{1}{2})$. 
In fact, integrating by parts leads to
\begin{align*}
\int_{s}^{\infty}xf(x)\mathrm{d}x &= \int_{s}^{\infty}\sqrt{\frac{x}{2\pi}}{\rm e}^{-\frac{x}{2}}\mathrm{d}x\nonumber \\
&= \sqrt{\frac{2s}{\pi}}{\rm e}^{-\frac{s}{2}} + \int_{s}^{\infty}f(x)\mathrm{d}x,
\end{align*}
while the change of variable yields
\begin{align*}
\int_{s}^{\infty}f(x)\mathrm{d}x &= \int_{s}^{\infty}\sqrt{\frac{1}{2\pi x}}{\rm e}^{-\frac{x}{2}}\mathrm{d}x \\
&=  \sqrt{\frac{2}{\pi s}}{\rm e}^{-\frac{s}{2}} \int_{0}^{\infty}\sqrt{\frac{s}{s+x}}\frac{1}{2}{\rm e}^{-\frac{x}{2}}\mathrm{d}x.
\end{align*}
From \eqref{eqn:chi-square-conditional-expectation} it is obvious that 
\begin{equation}\label{eqn:T-square-greater-than-t}
  T^2=\probE(X^2\big |X^2>t)>t
\end{equation}
and
\begin{equation}
\left|\probE\left(X^2 \big| X^2 > s_1\right) - \probE\left(X^2 \big| X^2 > s_2\right)\right|\le C|s_1 - s_2|
\end{equation}
for some positive constant $C$ when $s_1$, $s_2$ is bounded away from zero.
In particular, we have
\begin{align}
\left|\probE  T_k^2 - T^2\right| 
=& \left|\probE  \left(\probE  T_k^2 \big| X_{(k + 1)}^2\right) - T^2\right|\nonumber \\
=& \left|\probE  \left(\probE  T_k^2 \big| X_{(k + 1)}^2\right) - \probE  \left(T_k^2\big| X_{(k + 1)}^2 = t\right)\right|\nonumber \\
\le &C \probE |X_{(k + 1)}^2 - t| \\
\le & C\sqrt{\frac{t}{k\log \frac{n}{k}}},
\end{align}
where  the last inequality follows from \cite{Boucheron2012Concentration} and Lemma~\ref{lem:ord_concentration}.
This proves
\begin{equation}
    |\probE  T_k - T| \le \frac{C\sqrt t}{T\sqrt{k\log \frac{n}{k}}}.
\end{equation}
But \eqref{eqn:T-square-greater-than-t} implies $T>\sqrt t$, thus
\begin{equation}
  |\probE  T_k - T| \le \frac{C}{\sqrt{k\log \frac{n}{k}}}.
\end{equation}

The proof is concluded by invoking \eqref{eqn:T-concentration} to obtain
\begin{equation} 
    \probP\left( |T_k-T| > \frac{C}{\sqrt{k\log \frac{n}{k}}} + \varepsilon \right) \le 2{\rm e}^{-\frac{k\varepsilon^2}{2}}.
\end{equation}

% ------------------------------------------------------

\section{Discussions}\label{sec:discussion}
In this section we discuss the extension of Theorem~\ref{thm:lowerbound} 
to sub-Gaussian matrices 
and that of Proposition~\ref{prop:gaussian-order-stat} to general distributions. 
Revisiting the proof of Theorem~\ref{thm:lowerbound}, 
one finds that there are two places where being Gaussian is critical:
\begin{enumerate}
  \item Applying Proposition~\ref{prop:gaussian-order-stat}.
  \item Bounding $\sum_{i,j=2}^N v_i'v_j'\langle\vect b_i,\vect b_j\rangle$ by independence and Hanson-Wright inequality.
\end{enumerate}

To resolve these issues, 
we first discuss the extension of Proposition~\ref{prop:gaussian-order-stat} to non-Gaussian cases. 
It is evident that the use of concentration of measure in proving Proposition~\ref{prop:gaussian-order-stat} 
is redundant: 
one may instead try to bound $T_k-T$ and $\probE T_k-T$ by utilizing Lemma~\ref{lem:ord_concentration} only. 
The main difficulty arising in doing so is that \eqref{eqn:deviation-and-density} is no longer applicable 
and  $$T_k^2\Big|X_{(k+1)}\overset{(d)}{=}\frac 1k\sum Y_i$$ no longer holds 
for discontinuous distributions as in \eqref{eqn:ord-stat-condition-expectation}. 
However, if $F$ is close to chi-squared distribution, hence has small jumps, 
it is still possible to control the deviation probability in \eqref{eqn:deviation-and-density} by $\delta$. 
As for \eqref{eqn:ord-stat-condition-expectation}, 
we show a simple trick that reduces the situation to continous case. 
Let $G$ be the generalized inverse function of $F$, defined by 
\begin{equation*}
  G(x)=\inf\{t:F(t)\ge x\}
\end{equation*}
and let $U_1,\ldots,U_n$ be i.i.d. random variables following uniform distribution on $[0,1]$. 
Then 
\begin{equation*}
  \left(X_{(1)},\ldots,X_{(n)}\right)\overset{(d)}=\left(G\left(U_{(1)}\right),\ldots, G\left(U_{(n)}\right)\right).
\end{equation*}
Note that $G$ is monotonically increasing. 
Hence $G\left(U_{(k+1)}\right)=s$ is equivalent to $U_{(k+1)}\in[\alpha,\beta]$ for some $0\le\alpha\le\beta\le 1$. 
Thus $T_{(k)}^2\big|\{X_{(k+1)}=s\}$ has the same distribution as $\sum G^2\left(U_{(i)}\right)\big|\{U_{(k+1)}\in[\alpha,\beta]\}$. 
Since the distribution of $U_{(i)}$ is continuous, 
what we desire follows from controlling $\sum G^2\left(U_{(i)}\right)\big|U_{(k+1)}$ 
by the same argument in Proposition~\ref{prop:gaussian-order-stat} and integrating. 

In a word, we have sketched a proof that: 
\emph{Proposition~\ref{prop:gaussian-order-stat} holds with slight modification for distributions 
that are sufficiently close to chi-squared.} 
By central limit theorem, 
the distribution of $x_i$ defined after \eqref{eqn:Av} converges to Gaussian, 
and consequently the distribution of $x_i^2$ converges to $\chi_1^2$. 
This implies that Proposition~\ref{prop:gaussian-order-stat} remains valid for sub-Gaussian matrices, 
hence resolving the issue 1) above.

Issue 2) is more delicate and requires some modifications 
in the statement of Theorem~\ref{thm:lowerbound} to generalize to sub-Gaussian matrices. 
To avoid unnecessary technicalities, 
we will not explicitly state these modifications here, 
but instead explain the main strategy for such generalization as follows.
It follows from our construction that $(v_2',\ldots,v_n')$ is $(s-1)$-sparse, 
hence we have 
\begin{align*}
  \frac1n\left|(n-1)(1-v_1^2)-\sum_{i,j=2}^N v_i'v_j'\langle\vect b_i,\vect b_j\rangle\right|
  \le\left(2\delta_{s-1}(\vect B)+\delta_{s-1}^2(\vect B)\right)(1-v_1^2),
\end{align*}
where $$\vect B=\frac{1}{\sqrt n}[\vect b_2,\ldots,\vect b_N]\in\mathbb R^{n\times(N-1)}.$$ 
Note that $\vect B$ is a submatrix of $\vect P\vect A$, 
where $\vect P$ is the orthogonal projection onto the orthogonal complement of $\vect a_1$. 
It follows from, for instance, Cauchy interlacing law that $\delta_{s-1}(\vect B)\le\delta_s(\vect A)$. 
For sub-Gaussian random matrices $\vect A$, 
upper bounds of $\delta_s(\vect A)$ are well-known (Theorem~\ref{thm:sub-Gaussian-RIP-upper-bound}). 
Thus the term $\sum_{i,j=2}^N v_i'v_j'\langle\vect b_i,\vect b_j\rangle$ 
is of the same order as $T_k$ with overwhelming probability. 
To get rid of such weakness, one may take $v_1$ to be properly close to $1$, 
which yields faster decay of the term $(1-v_1^2)\delta_{s-1}(\vect B)$ than 
the decay of $\sqrt{1-v_1^2}\sqrt{\sum_{j=1}^{s-1}x_{(j)}^2}$, 
so the error term will be eventually dominated by $\sqrt{1-v_1^2}\sqrt{\sum_{j=1}^{s-1}x_{(j)}^2}$. 
The error term this approach gives is actually worse than that in Theorem~\ref{thm:lowerbound} by a constant, 
which should not be a serious concern in most applications, for instance, in Corollary~\ref{cor:no.-of-measurements}. 

% ------------------------------------------------------

\section{Numerical Experiments}\label{sec:simulations}
We present some numerical calculations in this section to 
provide an intuition of our new lower and upper bounds for the RIP constant 
described in Theorem~\ref{thm:lowerbound} and Proposition~\ref{prop:upperbound}.
We fix sparsity levels $s/N = (0.1, 0.01, 0.001), N = 10000$ and compute the $99\%$-confidence intervals 
for different compression rates, see Figure~\ref{fig:fig1}.
We find that when the compression rate $N/n$ is not too large, 
our new bounds are quite tight compared to previous ones.

\begin{figure}[t]
\centering
\centerline{\includegraphics[width=9.5cm]{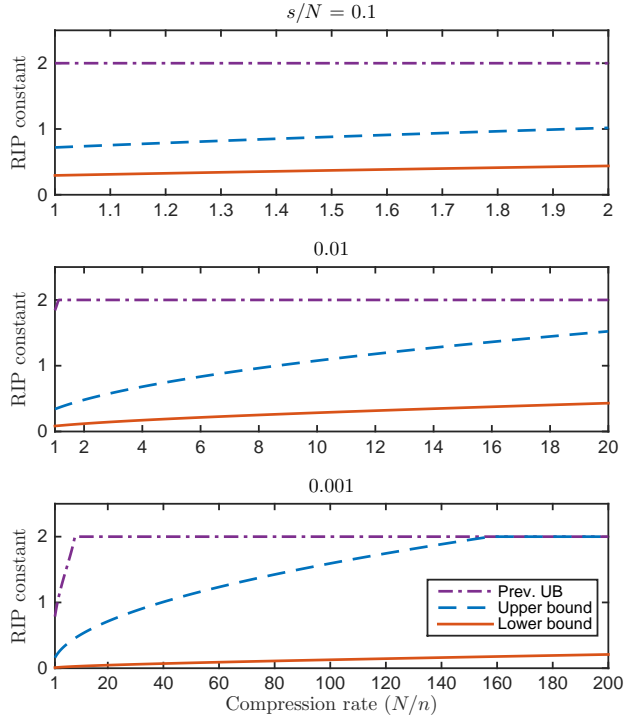}}
\caption{Plot of the lower bound and the upper bound of the RIP constant derived in this paper 
comparing with previous results. 
The $x$-coordinate is compression rate. 
The sparsity level $s/N$ is taken to be $0.1$, $0.01$, $0.001$ respectively in three figures.
Parameters $N$ are fixed to $1000$, and confidence level is $0.99$.
All curves are truncated above by $2$ for visual convenience.}
\label{fig:fig1}
\end{figure}

% ------------------------------------------------------

\section{Conclusion}\label{sec:conclusion}
In this paper we gave a lower bound of RIP constants for Gaussian random matrices 
and discussed the strategy to generalize this result to sub-Gaussian random matrices. 
For Gaussian random matrices, this bound was shown to be tight 
by comparing with a new upper bound of RIP constants
which is better than previous results by a multiplicative constants in most cases. 
The new lower bound of RIP constants implies the fact 
that the minimal number of sub-Gaussian measurements 
needed to enforce the RIP of order $s$ is $\Omega(s\log({\rm e}N/s))$, 
which was proved in literature by with a much more sophisticated approach. 
In our proof we also established a concentration inequality for sum of top order statistics, 
which we believe to be of interest in other fields besides compressed sensing.

% ------------------------------------------------------

\section{Appendix: Proof of Proposition~\ref{prop:asymp}}\label{app:asymp}
It is well-known that Gaussian tail bound is asymptotically equivalent to $\frac 2x{\rm e}^{-x^2/2}$:
\begin{equation}
  \probP(|X|>x)\sim\frac{2}{x}{\rm e}^{-x^2/2},
\end{equation}
Thus 
\begin{equation}
  \frac sN\sim\probP(X^2>t)\sim\frac{2}{\sqrt t}{\rm e}^{-t/2}.
\end{equation}
Taking logarithms on both sides, we obtain
\begin{equation}
  t+\log t\sim 2\log\frac Ns
\end{equation}
As $s/N\to\infty$, $t$ also tends to infinity, thus $\log t=o(t)$. 
It follows that
\begin{equation}
  t\sim 2\log\frac Ns.
\end{equation}

To prove $T\sim\sqrt{2\log\frac Ns}$, 
we will use \eqref{eqn:chi-square-conditional-expectation}. 
By bounded convergence theorem we see
\begin{equation*}
  \lim_{s\to\infty}\probE\sqrt{\frac{s}{s+Y}}=1.
\end{equation*}
This combined with \eqref{eqn:chi-square-conditional-expectation} implies
\begin{equation*}
  \probE(X^2\Big|X^2>t)\sim t,
\end{equation*}
hence
\begin{equation}
  T=\sqrt{\probE(X^2\Big|X^2>t)}\sim\sqrt t\sim\sqrt{2\log\frac Ns}.
\end{equation}

\section{Appendix: Tools from Probability Theory}\label{app:tool}

In this appendix we collect some tools from probability theory 
that play a role (but are not essential) in our proof. 
\begin{thm}[Concentration of measure, \cite{ledoux2001concentration}]
Let $X$ be a standard Gaussian vector taking values in $\mathbb R^n$ 
and $f:\mathbb R^n\to\mathbb R$ a $K$-Lipshitz function, i.e. 
\begin{equation*}
  |f(\vect x)-f(\vect y)|\le K\|\vect x-\vect y\|,\quad\text{for all $\vect x,\vect y\in\mathbb R^n$.}
\end{equation*}
Then we have
\begin{gather*}
  \probP(f(X)-\probE f(X)>t)\le\mathrm {\rm e}^{-\frac{t^2}{2K^2}},\\
  \probP(f(X)-\probE f(X)<-t)\le\mathrm {\rm e}^{-\frac{t^2}{2K^2}}.
\end{gather*}
\end{thm}

The next lemma is on comparison of Gaussian processes that plays an important role 
in proving Proposition~\ref{prop:upperbound}.
\begin{lem}[Slepian-Fernique lemma, \cite{gordon1987elliptically}]\label{lem:slepian-fernique}
Let $(X_{s,t})_{s\in S, t\in T}$ and $(Y_{s,t})_{s\in S, t\in T}$ be two Gaussian processes 
defined on the index set $S\times T$. 
Assume that for any $s,s'\in S$, $t,t'\in T$, $s\ne s'$ we have
\begin{gather*}
  \probE|X_{s,t}-X_{s,t'}|^2\le\probE|Y_{s,t}-Y_{s,t'}|^2,\\
  \probE|X_{s,t}-X_{s',t'}|^2\ge\probE|Y_{s,t}-Y_{s',t'}|^2.
\end{gather*}
Then 
\begin{equation*}
  \probE\min_{s\in S}\max_{t\in T}X_{s,t}\le\probE\min_{s\in S}\max_{t\in T}Y_{s,t}.
\end{equation*}
\end{lem}
\begin{cor}\label{cor:slepian-fernique}
Let $(X_t)_{t\in T}$ and $(Y_t)_{t\in T}$ be two Gaussian processes 
defined on the index set $T$. 
Assume that for any $t,t'\in T$ we have
\begin{equation*}
  \probE|X_t-X_{t'}|^2\le\probE|Y_t-Y_{t'}|^2.
\end{equation*}
Then 
\begin{equation*}
  \probE\max_{t\in T}X_t\le\probE\max_{t\in T}Y_t.
\end{equation*}
\end{cor}

\begin{proof-special}
  We first check that 
  \begin{equation}\label{eqn:app1}
    \probE|X_{\vect u,\vect v}-X_{\vect u',\vect v}|^2=\probE|Y_{\vect u,\vect v}-Y_{\vect u',\vect v}|^2. 
  \end{equation}
  In fact, the left hand side is equal to $\|\vect u-\vect u'\|^2\|\vect v\|^2$, 
  while the right hand is equal to $\|\vect u-\vect u'\|^2$. 
  This implies \eqref{eqn:app1} since $\|v\|=1$ by definition of $\Psi_k$. 

  Next we show 
  $$
    \probE|X_{\vect u,\vect v}-X_{\vect u',\vect v'}|^2\le\probE|Y_{\vect u,\vect v}-Y_{\vect u',\vect v'}|^2, 
  $$
  for \emph{all} $\vect u,\vect u'\in\mathbb S^{n-1}$, $\vect v,\vect v'\in\Psi_k$. 
  A little computation shows this is equivalent to
  \begin{equation}\label{eqn:app2}
    \sum_i\sum_j|u_iv_j-u_i'v_j'|^2\le\|\vect u-\vect u'\|^2+\|\vect v-\vect v'\|^2.
  \end{equation}
  By expanding $|u_iv_j-u_i'v_j'|^2$ as 
  \begin{equation*}
    (u_i-u_i')^2v_j^2+u_i'^2(v_j-v_j')^2+2(u_i-u_i')v_ju_i'(v_j-v_j'),
  \end{equation*}
  equation \eqref{eqn:app2} is reduced to 
  $$
    (1-\langle\vect u,\vect u'\rangle)(1-\langle\vect v,\vect v'\rangle)\ge0,
  $$
  which readily follows from Cauchy-Schwarz inequality.

  The upper bound in Lemma~\ref{lem:slepian} is an easy corollary of \eqref{eqn:app2} 
  and Corollary~\ref{cor:slepian-fernique}, 
  while the lower bound follows from \eqref{eqn:app1}, \eqref{eqn:app2} and Lemma~\ref{lem:slepian-fernique}.
\end{proof-special}

% ------------------------------------------------------

\bibliographystyle{IEEEtran}
\bibliography{IEEEabrv,mybibfile}
\vfill\pagebreak
%\clearpage

\end{document}